\batchmode
\makeatletter
\def\input@path{{/Users/axelaraneda/Desktop/Research/fBM/}}
\makeatother
\documentclass[english]{article}
\usepackage{lmodern}
\usepackage[T1]{fontenc}
\usepackage[latin9]{inputenc}
\usepackage[a4paper]{geometry}
\usepackage{amsmath}
\usepackage{amsthm}
\usepackage{amssymb}
\usepackage{graphicx}

\makeatletter
\newcommand{\lyxaddress}[1]{
	\par {\raggedright #1
	\vspace{1.4em}
	\noindent\par}
}
\theoremstyle{remark}
\newtheorem{rem}{\protect\remarkname}[section]
\theoremstyle{plain}
\newtheorem{lem}{\protect\lemmaname}[section]
\theoremstyle{plain}
\newtheorem{thm}{\protect\theoremname}[section]
\theoremstyle{definition}
\newtheorem{defn}{\protect\definitionname}[section]
\theoremstyle{plain}
\newtheorem{cor}{\protect\corollaryname}[section]

\usepackage{hyperref}
\usepackage[numbers,sort&compress]{natbib}
\date{}

\usepackage{todonotes}

\@ifundefined{showcaptionsetup}{}{%
 \PassOptionsToPackage{caption=false}{subfig}}
\usepackage{subfig}
\makeatother

\usepackage{babel}
\providecommand{\corollaryname}{Corollary}
\providecommand{\definitionname}{Definition}
\providecommand{\lemmaname}{Lemma}
\providecommand{\remarkname}{Remark}
\providecommand{\theoremname}{Theorem}

\begin{document}
\title{\textbf{The sub-fractional CEV model }}
\author{Axel A.~Araneda\textsuperscript{a,}\thanks{Corresponding author. Email: \protect\href{mailto:axelaraneda@mail.muni.cz}{\texttt{axelaraneda@mail.muni.cz}}}
\and Nils Bertschinger\textsuperscript{b,c,}\thanks{Email: \protect\href{mailto:bertschinger@fias.uni-frankfurt.de}{\texttt{bertschinger@fias.uni-frankfurt.de}}}}
\maketitle

\lyxaddress{\begin{center}
\vspace{-2em}$\textsuperscript{a}$ Institute of Financial Complex
Systems \\Faculty of Economics and Administration\\ Masaryk University\\
602 00 Brno, Czech Republic.
\par\end{center}}

\lyxaddress{\begin{center}
\vspace{-2em}$\textsuperscript{b}$ Frankfurt Institute for Advanced
Studies\\60438 Frankfurt am Main, Germany.
\par\end{center}}

\lyxaddress{\begin{center}
\vspace{-2em}$\textsuperscript{c}$ Department of Computer Science\\
Goethe University\\ 60629 Frankfurt am Main, Germany.
\par\end{center}}

\begin{center}
\vspace{-1em} This version: March 24, 2021 \vspace{1.5em}
\par\end{center}
\begin{abstract}
The sub-fractional Brownian motion (sfBm) is a stochastic process,
characterized by non-stationarity in their increments and long-range
dependency, considered as an intermediate step between the standard
Brownian motion (Bm) and the fractional Brownian motion (fBm). The
mixed process, a linear combination between a Bm and an independent
sfBm, called mixed sub-fractional Brownian motion (msfBm), keeps the
features of the sfBm adding the semi-martingale property for $H>3/4$,
is a suitable candidate to use in price fluctuation modeling, in particular
for option pricing. In this note, we arrive at the European Call price
under the Constant Elasticity of Variance (CEV) model driven by a
mixed sub-fractional Brownian motion. Empirical tests show the capacity
of the proposed model to capture the temporal structure of option
prices across different maturities.


\vspace{2em}
\end{abstract}

\section{Introduction}

The sub-fractional Brownian motion, in short sfBm, is a stochastic
process that emerges from the occupation time fluctuations of branching
particle systems\footnote{It also arises in systems without branching \cite{bojdecki2006limit,bojdecki2010particle}.
Besides, sfBm appears independently as an expansion of the fBm \cite{dzhaparidze2004series}.} \cite{bojdecki2004sub}. It owns the main properties of the fractional
Brownian motion\footnote{A fBm, $B^{H}=\left\{ B_{t}^{H},t\geq0\right\} $, is a Gaussian process
with $0<H<1$ , characterized by:

\quad{}i) $B_{0}^{H}=0$,

\quad{}ii) $\mathbb{E}\left(B_{t}^{H}\right)=0$,

\quad{}iii) $\mathbb{E}\left(B_{t}^{H}\cdot B_{s}^{H}\right)=\frac{1}{2}\left\{ \left|t\right|^{2H}+\left|s\right|^{2H}-\left|t-s\right|^{2H}\right\} $
.} (fBm) \cite{mandelbrot1968fractional} as long-range dependence,
self-similarity and Holder paths. However, a key difference among
them is the sfBm has non-stationary increments. Besides, the sfBm
has more weakly correlated increments and their covariance decays
at a higher rate, in comparison to the fBm. See \cite{bojdecki2004sub,tudor2007some,mishura2018stochastic}
for details and properties of the sfBm.

Then, diffusion processes under sfBm could be considered as a potential
candidate to model some financial time-series which exhibit long-range
dependency and non-stationarity increments \cite{bassler2007nonstationary,morales2013non,costa2003long,baldovin2011modeling}.
However, the option pricing under sfBm leads to arbitrage opportunities
\cite{zhang2017arbitrage}. In this context, the mixed sub-fractional
Brownian motion (msfBm) \cite{charles2015sub}, a linear combination
between a standard Brownian motion (Bm) and an independent sfBm, which
holds the main properties of the latter but adds the semi-martingale
condition when the Hurst exponent $H\in(3/4,1)$; emerges as a suitable
alternative for financial modeling.

Previously, some attempts have been addressed in the literature to
include sub-fractional diffusion in the price fluctuation modeling
\cite{tudor2008sub2,liu2010sub,xu2018pricing,xu2019pricing}; mainly,
using the seminal Black-Scholes (B-S) model as a platform \cite{black1973pricing}.
However, in this communication, we will consider a sub-fractional
extension to the Constant Elasticity of Variance (CEV) model \cite{Cox1975,cox1996constant},
which is capable to address some shortcomings of the B-S approach
as the leverage effect and the implied volatility skew \cite{eltit}. 

Following the procedure given in \cite{araneda2020fractional} to
the fractional and mixed-fractional case, we derive the Fokker-Planck
equation under sub-fractional diffusion and the transition probability
density function for the mixed sub-fractional CEV (msfCEV) is obtained,
leading to the price formula for an European Call option in terms
of the non-central chi-squared distribution and the M-Whittaker function. 

The outline of the paper is the following. First, some properties
and details of the dfBm are addressed, including their stochastic
calculus rules and the Fokker-Planck equation. Later, we put the accent
on the msfBm. At next, the option pricing under the standard CEV model
is reviewed. In section \ref{sec:The-msfCEV-model}, the principal
results for the msfCEV option pricing model are presented. At next,
some numerical results using real option pricing data are delivered.
Finally, the main conclusions are displayed.

\section{On the sub-fractional Brownian motion}

The sub-fractional Brownian motion is centered Gaussian process $\xi^{H}=\left\{ \xi_{t}^{H},t\geq0\right\} $
, with $\xi_{0}^{H}=0$ and covariance:

\begin{equation}
\mathbb{E}\left(\xi_{t}^{H}\cdot\xi_{s}^{H}\right)=s^{2H}+t^{2H}-\frac{1}{2}\left[\left(s+t\right)^{2H}+\left|t-s\right|^{2H}\right]\label{eq:cov}
\end{equation}

\noindent being $H\in\left]0,1\right[$ the Hurst index. For $H=1/2,$
$\xi^{H}$ becomes a standard Brownian motion.

We highlight two main properties of the sfBm derived from Eq. (\ref{eq:cov}),
in the following remarks.
\begin{rem}[Correlation of the increments for sfBm]
Let $t>s\geq v>u\geq0.$ The covariance over non-overlapping increments
is given by:

\begin{eqnarray*}
\mathbb{E}\left[\left(\xi_{v}^{H}-\xi_{u}^{H}\right)\cdot\left(\xi_{t}^{H}-\xi_{s}^{H}\right)\right] & = & \frac{1}{2}\left[\left(t+u\right)^{2H}+\left(t-u\right)^{2H}+\left(s+v\right)^{2H}+\left(s-v\right)^{2H}\right.\\
 &  & \left.-\left(t+v\right)^{2H}-\left(t-v\right)^{2H}-\left(s+u\right)^{2H}-\left(s-u\right)^{2H}\right]
\end{eqnarray*}

Thus, $\mathbb{E}\left[\left(\xi_{v}^{H}-\xi_{u}^{H}\right)\cdot\left(\xi_{t}^{H}-\xi_{s}^{H}\right)\right]>0$
(resp. <0) for $H>\frac{1}{2}$ (resp. $H<\frac{1}{2}$), which implies
positive (resp. negative) correlated increments.\\
\end{rem}
\begin{rem}[The sfBM has non-stationarity increments]
\label{rem:non-stationarity}The variance for a general increment
is given by:

\[
\mathbb{E}\left[\left|\xi_{t}^{H}-\xi_{s}^{H}\right|^{2}\right]=-2^{2H-1}\left(s^{2H}+t^{2H}\right)+\left(s+t\right)^{2H}-\left(t-s\right)^{2H},\qquad0\leq s<t
\]

Then, for any $u\in\mathbb{\mathbb{R^{+}}}$; is clear that: 

\[
\mathbb{E}\left[\left|\xi_{t}^{H}-\xi_{s}^{H}\right|^{2}\right]\neq\mathbb{E}\left[\left|\xi_{t+u}^{H}-\xi_{s+u}^{H}\right|^{2}\right]
\]
\\
\end{rem}
On the other hand, It's well known that $\xi^{H}$ is neither a Markov
process nor a semi-martingale for $H\neq1/2$. So, the classical Ito
rules are not available for sfBm with $H\neq1/2$. Since that, the
following results are useful for our purposes.
\begin{lem}[It\^o\textquoteright s lemma for a sub-fractional Brownian motion]
\label{lem: ito} Let $f=f(\xi_{t}^{H})\in C^{2}(\mathbb{R})$ and
$\frac{1}{2}\leq H<1$. Then:

\[
\text{d}f=f'(\xi_{t}^{H})\text{d}\xi_{t}^{H}+Ht^{2H-1}\left(2-2^{2H-1}\right)f''(\xi_{t}^{H})\text{d}t
\]
\end{lem}
\begin{proof}
See \cite{yan2011ito}.
\end{proof}
\begin{thm}
\label{thm:The-Fokker-Planck-equation}The Fokker-Planck equation
related to the generic process 

\begin{equation}
\text{d}y_{t}=\mu\left(y_{t},t\right)\text{d}t+\sigma\left(y_{t},t\right)\text{d\ensuremath{\xi_{t}^{H}}}\label{eq:dy}
\end{equation}

is given by:

\begin{equation}
\frac{\partial P}{\partial t}=Ht^{2H-1}\left(2-2^{2H-1}\right)\frac{\partial\left(\sigma^{2}P\right)}{\partial y^{2}}-\frac{\partial\left(\mu P\right)}{\partial y}\label{eq:FPE}
\end{equation}
\\
\end{thm}
\begin{proof}
We will follow and extend the procedure given in \cite{unal2007fokker}
for diffusion processes under fractional Brownian motion, but this
time applied to sub-fractional case.

Let $g=g\left(y_{t}\right)$ a twice differentiable scalar function.
Using the It\^o formula for sfBm (Lemma \ref{lem: ito}), we have:

\begin{equation}
\text{d}g=\left[\mu\frac{\partial g}{\partial y}+Ht^{2H-1}\left(2-2^{2H-1}\right)\sigma^{2}\frac{\partial^{2}g}{\partial y^{2}}\right]\mathrm{d}t+\sigma\frac{\partial h}{\partial y}\mathrm{d}\xi_{t}^{H}\label{eq:dh}
\end{equation}
\\

Then, taking expectations over (\ref{eq:dh}):

\begin{equation}
\frac{\text{d}\mathbb{E}\left(g\right)}{\mathrm{d}t}=\mathbb{E}\left(\mu\frac{\partial g}{\partial y}\right)+\mathbb{E}\left[Ht^{2H-1}\left(2-2^{2H-1}\right)\sigma^{2}\frac{\partial^{2}g}{\partial y^{2}}\right]\label{eq:Edxdt}
\end{equation}
\\

Since the expectation of $g(y_{t})$ is:

\begin{equation}
\mathbb{E}\left[g\left(y_{t}\right)\right]=\int g\left(y\right)P\left(y,t\right)\text{d}y\label{eq:Eh}
\end{equation}
\\

\noindent where $P$ is the transition probability density function
at time $t$; the relations (\ref{eq:Edxdt}) and (\ref{eq:Eh}) yields
to:

\begin{equation}
\int_{-\infty}^{\infty}g\frac{\partial P}{\partial t}\text{d}y=\int_{-\infty}^{\infty}\left[\mu\frac{\partial g}{\partial y}+Ht^{2H-1}\left(2-2^{2H-1}\right)\sigma^{2}\frac{\partial^{2}g}{\partial y^{2}}\right]P\text{d}y\label{eq:integral-FP}
\end{equation}
\\

After that, using the following results:

\[
\int_{-\infty}^{\infty}\mu\frac{\partial g}{\partial y}P\text{d}y=-\int_{-\infty}^{\infty}g\frac{\partial\left(\mu P\right)}{\partial y}\text{d}y
\]

\[
\int_{-\infty}^{\infty}\sigma^{2}\frac{\partial^{2}g}{\partial y^{2}}P\text{d}y=\int_{-\infty}^{\infty}g\frac{\partial\left(\sigma^{2}P\right)}{\partial y^{2}}\text{d}y
\]
\\

\noindent the Eq. (\ref{eq:integral-FP}) goes to:

\begin{equation}
\int_{-\infty}^{\infty}h\left[\frac{\partial P}{\partial t}+\frac{\partial\left(\mu P\right)}{\partial y}-Ht^{2H-1}\left(2-2^{2H-1}\right)\frac{\partial\left(\sigma^{2}P\right)}{\partial y^{2}}\right]\text{d}x=0\label{eq:infp}
\end{equation}
\\

Finally, the Fokker-Planck equation related to the process (\ref{eq:dy}),
emerges from (\ref{eq:infp}).
\end{proof}

\section{On the mixed Sub-fractional Brownian motion}

Similarly to approach followed for Cheridito \cite{cheridito2001mixed}
to the construction of a mixed fractional Brownian motion, a msfBm
was introduced \cite{tudor2007some,charles2015sub} as a linear combination
among a sfBm and an ordinary \& independent Brownian motion ($B_{t}$).
\begin{defn}
Let $M^{\beta,\gamma,H}=\left\{ M_{t}^{\beta,\gamma,H},t\geq0\right\} $
a msfBm, defined by:

\begin{equation}
M_{t}^{\beta,\gamma,H}=\beta B_{t}+\gamma\xi_{t}^{H};\qquad\beta\geq0,\,\gamma\geq0,\label{eq:M}
\end{equation}
\end{defn}
From \ref{eq:M} is clear that $M_{t}^{0,1,H}$ becomes a sfBM; while
$M_{t}^{1,0,H}$ and $M_{t}^{0,1,1/2}$ represents a standard Bm.
Besides, from Eq. (\ref{eq:cov}), we have:

\begin{equation}
\mathbb{E}\left(M_{t}^{\beta,\gamma,H}\cdot M_{s}^{\beta,\gamma,H}\right)=\beta^{2}\min\left(s,t\right)+\gamma^{2}\left\{ s^{2H}+t^{2H}-\frac{1}{2}\left[\left(s+t\right)^{2H}+\left|t-s\right|^{2H}\right]\right\} \label{eq:covarM}
\end{equation}

As in the previous section, the correlation over non-overlapping increments
and the non-stationarity increments property are addressed at next.
\begin{rem}[Correlation of the increments for msfBm]
\label{rem:correlation}Let $t>s\geq v>u\geq0.$ The covariance over
non-overlapping increments, for the process $M^{\beta,\gamma,H}$
is given by:

\begin{eqnarray*}
\mathbb{E}\left[\left(M_{v}^{\beta,\gamma,H}-M_{u}^{\beta,\gamma,H}\right)\cdot\left(M_{t}^{\beta,\gamma,H}-M_{s}^{\beta,\gamma,H}\right)\right] & = & \frac{\gamma^{2}}{2}\left[\left(t+u\right)^{2H}+\left(t-u\right)^{2H}+\left(s+v\right)^{2H}+\left(s-v\right)^{2H}\right.\\
 &  & \left.-\left(t+v\right)^{2H}-\left(t-v\right)^{2H}-\left(s+u\right)^{2H}-\left(s-u\right)^{2H}\right]
\end{eqnarray*}

The RHS of the above equation is positive for $H>\frac{1}{2}$, which
indicates that the increment are correlated on that region. Alternatively,
for $H<\frac{1}{2}$ the increments are negative correlated.\\
\end{rem}
\begin{rem}[The msfBM has non-stationarity increments]
\label{rem:non-stationarity-M}. Let $\left(s,t,u\right)\in\mathbb{R}_{+}^{3}$
and $t>s$. Then:

\[
\mathbb{E}\left[\left|M_{t}^{\beta,\gamma,H}-M_{s}^{\beta,\gamma,H}\right|^{2}\right]=\beta^{2}\left(t-s\right)+\gamma^{2}\left[-2^{2H-1}\left(s^{2H}+t^{2H}\right)+\left(s+t\right)^{2H}-\left(t-s\right)^{2H}\right]
\]

Thus, is easy to check the non-stationarity property for the increments
by:

\[
\mathbb{E}\left[\left|M_{t}^{\beta,\gamma,H}-M_{s}^{\beta,\gamma,H}\right|^{2}\right]\neq\mathbb{E}\left[\left|M_{t+u}^{\beta,\gamma,H}-M_{s+u}^{\beta,\gamma,H}\right|^{2}\right]
\]
\\
\end{rem}
Another property, key for the msfBm process (\ref{eq:M}), and especially
useful in mathematical finance, is the semi-martingale property detailed
in the following statement.
\begin{lem}[Semi-martingale property for the msfBm]
\label{lem:-The-mixed} The mixed process $M_{t}^{\beta,\gamma,H}$,
with $\beta\neq0$ and $\frac{3}{4}<H<1$, is a semi-martingale equivalent
in law to $\beta\times B_{t}$.
\end{lem}
\begin{proof}
cf. \cite{tudor2007some,charles2015sub}.\\

The results provided in the Remarks \ref{rem:correlation}-\ref{rem:non-stationarity-M}
and the Lemma \ref{lem:-The-mixed}, makes the stochastic processes
driven by msfBm suitable for the option pricing modeling of time series
which presents non-stationarity and autocorrelation. In addition,
the following auxiliary results are presented.
\end{proof}
\begin{cor}
\label{cor:Let--a}Let $y$ a stochastic process ruled by:

\begin{equation}
\text{d}y_{t}=\mu(y_{t},t)\text{d}t+\sigma(y_{t},t)\text{d}M_{t}^{\beta,\gamma,H}\label{eq:dytilde}
\end{equation}

\noindent  and $h=h(y_{t})\in C^{2}(\mathbb{R})$. Then, we have:

\begin{eqnarray*}
dh & = & \left\{ \mu\frac{\partial h}{\partial y}+\left[\frac{1}{2}\beta^{2}+\gamma^{2}Ht^{2H-1}\left(2-2^{2H-1}\right)\right]\sigma^{2}\frac{\partial^{2}h}{\partial y^{2}}\right\} \mathrm{d}t+\sigma\frac{\partial h}{\partial y}\left(\beta\mathrm{d}B_{t}+\gamma\mathrm{d}\xi_{t}^{H}\right)\\
 & = & \left\{ \mu\frac{\partial h}{\partial y}+\left[\frac{1}{2}\beta^{2}+\gamma^{2}Ht^{2H-1}\left(2-2^{2H-1}\right)\right]\sigma^{2}\frac{\partial^{2}h}{\partial y^{2}}\right\} \mathrm{d}t+\sigma\frac{\partial h}{\partial y}\mathrm{d}M_{t}^{\beta,\gamma,H}
\end{eqnarray*}
\end{cor}
\begin{proof}
By direct consequence of the Lemma \ref{lem: ito}.
\end{proof}
\begin{cor}
\label{cor:The-Fokker-Planck-equation}The Fokker-Planck equation
for the drift mixed sub-fractional process described by the SDE (\ref{eq:dytilde})
is:

\[
\frac{\partial P}{\partial t}=\left[\frac{1}{2}\beta^{2}+\gamma^{2}Ht^{2H-1}\left(2-2^{2H-1}\right)\right]\frac{\partial\left(\sigma^{2}P\right)}{\partial y^{2}}-\frac{\partial\left(\mu P\right)}{\partial y}
\]
\end{cor}
\begin{proof}
Straightforward from Corollary \ref{cor:Let--a} and Theorem \ref{thm:The-Fokker-Planck-equation}.
\end{proof}

\section{The CEV model}

The constant elasticity of variance model assumes that asset price,
$S$, obeys the following It\^o drift-diffusion (under the risk-neutral
measure):

\begin{equation}
\text{d}S_{t}=rS_{t}\text{d}t+\sigma_{0}S_{t}^{\frac{\alpha}{2}}\text{d}B_{t}\label{eq:CEV}
\end{equation}

\noindent being $r\geq0$ the constant risk-free rate of interest,
$\sigma_{0}>0$ and $\alpha\in[0,2[$ the parameters of the model.
The selection of that interval for $\alpha$ ensure two desired properties
for the model:
\begin{enumerate}
\item The capacity to address the `leverage effect' and the skew phenomena
\cite{eltit}.
\item The arbitrage possibilities are neglected (existence of an unique
equivalent martingale measure) \cite{delbaen2002note}. 
\end{enumerate}
Defining $x=S^{2-\alpha}$, the It\^o's lemma yields to:

\[
\mathrm{d}x_{t}=\left(2-\alpha\right)\left[rx_{t}+\frac{1}{2}\left(1-\alpha\right)\sigma_{0}^{2}\right]\mathrm{d}t+\left(2-\alpha\right)\sigma_{0}\sqrt{x_{t}}\mathrm{d}B_{t}
\]

Later the transition density probability function $P(x_{T},T|x_{0},0)$,
follows the Fokker-Planck equation:

\begin{equation}
\frac{\partial P}{\partial t}=\frac{1}{2}\frac{\partial^{2}}{\partial x^{2}}\left[\left(2-\alpha\right)^{2}\sigma_{0}^{2}xP\right]-\frac{\partial}{\partial x}\left[\left(2-\alpha\right)\left(rx+\frac{1}{2}\left(1-\alpha\right)\sigma_{0}^{2}\right)P\right]\label{eq:FPCEV}
\end{equation}

The constant-coefficients parabolic equation (\ref{eq:FPCEV}) is
solved using the arguments supplied by Feller \cite{feller1951two}
(Feller's lemma, see also \cite{hsu2008constant}), where the transition
density in terms of the original variables is:

\[
P\left(\left.S_{T},T\right|S_{0},0\right)=\left(2-\alpha\right)k^{\frac{1}{2-\alpha}}\left(yw^{1-2\alpha}\right)^{\frac{1}{2\left(2-\alpha\right)}}\text{e}^{-y-w}I_{1/\left(2-\alpha\right)}\left(2\sqrt{yw}\right)
\]

\noindent with:

\begin{eqnarray}
k & = & \frac{2r}{\sigma_{0}^{2}\left(2-\alpha\right)\left[\text{e}^{r\left(2-\alpha\right)T}-1\right]},\label{eq:k}\\
y & = & kS_{0}^{2-\alpha}\text{e}^{r\left(2-\alpha\right)T},\label{eq:y}\\
w & = & kS_{T}^{2-\alpha}.\label{eq:subst}
\end{eqnarray}

Finally, the European Call option pricing, with maturity $T$ and
strike $E$, is obtained taking the expectation of the discounted
payoff, arriving to \cite{schroder1989computing}:

\begin{eqnarray*}
C\left(S_{0}\right) & = & \text{e}^{-rT}\int_{-\infty}^{\infty}\max\left\{ S_{T}-E,0\right\} P\left(\left.S_{T},T\right|S_{0},0\right)\text{d}S_{T}\\
 & = & S_{0}Q\left(2z;2+\frac{2}{2-\alpha},2y\right)-E\text{e}^{-rT}\left[1-Q\left(2y;\frac{2}{2-\alpha},2z\right)\right]
\end{eqnarray*}

\noindent being $Q\left(\bullet\right)$ the non-central chi-squared
complementary distribution function and $z=kE^{2-\alpha}$.

\section{The CEV model driven by a mixed sub-fractional Brownian motion\label{sec:The-msfCEV-model}}

Now, we consider that the asset price $S$ is ruled by following stochastic
differential equation:

\begin{equation}
\mathrm{d}S_{t}=rS_{t}\mathrm{d}t+\sigma S^{\frac{\alpha}{2}}\mathrm{d}M_{t}^{\beta,\gamma,H}\label{eq:SCEV}
\end{equation}
\noindent where the same conditions on the parameters as in the process
(\ref{eq:CEV}) are applied here. In addition, we restrict $\frac{1}{2}<H<1$.

Under the change of variable $x=S^{2-\alpha}$, and using the Corollary
(\ref{cor:Let--a}), the SDE (\ref{eq:SCEV}) goes to:

\begin{equation}
\mathrm{d}x_{t}=\left(2-\alpha\right)\left\{ rx_{t}+\left[\frac{1}{2}\beta^{2}+\gamma^{2}Ht^{2H-1}\left(2-2^{2H-1}\right)\right]\right\} \mathrm{d}t+\left(2-\alpha\right)\sigma\sqrt{x_{t}}\mathrm{d}M_{t}^{\beta,\gamma,H}\label{eq:dx}
\end{equation}

The evolution from $x(t=0)=x_{0}$ to $x(t=T)=x_{T}$ is given by
the transition probability density function $P=P(x_{T},T|x_{0},0)$.
Since By the Corollary (\ref{cor:The-Fokker-Planck-equation}), we
have:

\begin{multline}
\frac{\partial P}{\partial t}=\frac{\partial}{\partial x^{2}}\left\{ \left[\frac{1}{2}\beta^{2}+\gamma^{2}Ht^{2H-1}\left(2-2^{2H-1}\right)\right]\left(2-\alpha\right)^{2}\sigma^{2}xP\right\} \\
-\frac{\partial}{\partial x}\left\{ \left[rx+\left(\frac{1}{2}\beta^{2}+\gamma^{2}Ht^{2H-1}\left(2-2^{2H-1}\right)\right)\left(1-\alpha\right)\sigma^{2}\right]\left(2-\alpha\right)P\right\} \label{eq:FP}
\end{multline}

\begin{thm}
The transition probability density function related to the process
(\ref{eq:SCEV}) is written as:

\begin{eqnarray}
P\left(\left.S_{T},T\right|S_{0},0\right) & = & \left(2-\alpha\right)k_{s}^{\frac{1}{2-\alpha}}\left(y_{s}w_{s}^{1-2\alpha}\right)^{\frac{1}{2\left(2-\alpha\right)}}\text{e}^{-y-w}I_{1/\left(2-\alpha\right)}\left(2\sqrt{y_{s}w_{s}}\right)\label{eq:PH-1}
\end{eqnarray}
\\

\noindent with:

\begin{eqnarray*}
k_{s} & = & \left[\Phi\left(-\left(2-\alpha\right)rT\right)\right]^{-1},\\
y_{s} & = & k_{S}S_{0}^{2-\alpha}\text{e}^{r\left(2-\alpha\right)T}\\
w_{s} & = & k_{S}S_{T}^{2-\alpha}
\end{eqnarray*}

\noindent and

\begin{eqnarray*}
\Phi(T) & = & \gamma^{2}\sigma^{2}\frac{\left(2-\alpha\right)^{2}}{2H+1}\left(T\right)^{2H}\left(1-2^{2H-2}\right)\left\{ 2H+1+\text{e}^{\frac{1}{2}\left(2-\alpha\right)rT}\left[\left(2-\alpha\right)rT\right]^{-H}M_{H,H+1/2}\left[\left(2-\alpha\right)rT\right]\right\} \\
 &  & +\beta^{2}\frac{\sigma^{2}}{2r}\left(2-\alpha\right)\left(\text{e}^{\left(2-\alpha\right)rT}-1\right)
\end{eqnarray*}
\end{thm}
\begin{proof}
Defining the rescaled time $\tau=-\left(2-\alpha\right)rt$, the relation
(\ref{eq:FP}) becomes:

\[
\frac{\partial P}{\partial\tau}=\frac{\partial^{2}}{\partial x^{2}}\left[a(\tau)xP\right]+\frac{\partial}{\partial x}\left[\left(x+b\left(\tau\right)\right)P\right]
\]
\\

\noindent with, 

\begin{eqnarray*}
a(\tau) & = & -\frac{\sigma^{2}}{r}\left(2-\alpha\right)\left[\frac{1}{2}\beta^{2}+\gamma^{2}\left(1-2^{2H-2}\right)2H\left(-\frac{\tau}{\left(2-\alpha\right)r}\right)^{2H-1}\right]\\
b(\tau) & = & -\frac{\sigma^{2}}{r}\left(1-\alpha\right)\left[\frac{1}{2}\beta^{2}+\left(1-2^{2H-2}\right)2H\left(-\frac{\tau}{\left(2-\alpha\right)r}\right)^{2H-1}\right]
\end{eqnarray*}
\\

Since the ratio $a(\tau)/b(\tau)$ is time-independent, the PDE (\ref{eq:FP})
could be solved using the Feller's lemma with time varying coefficients
\cite{masoliver2016nonstationary,araneda2020fractional}. Thus:

\[
P\left(\left.x,\tau\right|x_{0},0\right)=\frac{1}{\phi(\tau)}\left(\frac{x\text{e}^{\tau}}{x_{0}}\right)^{{\textstyle {\displaystyle \frac{b-a}{2a}}}}\exp\left[-\frac{\left(x+x_{0}\text{e}^{-\tau}\right)}{\phi(\tau)}\right]I_{1-b/a}\left[\frac{2}{\phi(\tau)}\sqrt{\text{e}^{-\tau}x_{0}x}\right]
\]
\\
\noindent where $I_{\nu}$ is the modified Bessel function of first
kind of order $\nu$, and $\phi$ is defined by:

\begin{eqnarray*}
\phi(\tau) & = & -\frac{\sigma^{2}}{r}\left(2-\alpha\right)\int_{0}^{\tau}\left[\frac{1}{2}\beta^{2}+2\gamma^{2}\left(1-2^{2H-2}\right)H\left(\frac{s-\tau}{\left(2-\alpha\right)r}\right){}^{2H-1}\right]\text{e}^{-s}\text{d}s\\
 & = & \gamma^{2}\sigma^{2}\frac{\left(2-\alpha\right)^{2}}{2H+1}\left(-\frac{\tau}{b}\right)^{2H}\left(1-2^{2H-2}\right)\left[2H+1+\text{e}^{-\frac{1}{2}\tau}\left(-\tau\right)^{-H}M_{H,H+1/2}\left(-\tau\right)\right]\\
 &  & +\beta^{2}\sigma^{2}\frac{\text{\ensuremath{\left(2-\alpha\right)}}}{r}\left(\text{e}^{-\tau}-1\right)
\end{eqnarray*}
\\

\noindent being $M_{\kappa,\upsilon}\left(l\right)$ the M-Whittaker
function \cite{NIST}.

Later, in terms of the original variables $(S,t)$, the transition
probability density related to the process (\ref{eq:SCEV}) is written
as:

\begin{eqnarray}
P\left(\left.S_{T},T\right|S_{0},0\right) & = & \left(2-\alpha\right)k_{s}^{\frac{1}{2-\alpha}}\left(y_{s}w_{s}^{1-2\alpha}\right)^{\frac{1}{2\left(2-\alpha\right)}}\text{e}^{-y-w}I_{1/\left(2-\alpha\right)}\left(2\sqrt{y_{s}w_{s}}\right)\label{eq:PH}
\end{eqnarray}
\\

\noindent with:

\begin{eqnarray*}
y_{s} & = & k_{S}S_{0}^{2-\alpha}\text{e}^{r\left(2-\alpha\right)T},\\
w_{s} & = & k_{S}S_{T}^{2-\alpha}\\
k_{s} & = & \left[\phi\left(-\left(2-\alpha\right)rT\right)\right]^{-1}\\
 & = & \left[\Phi\left(T\right)\right]^{-1}
\end{eqnarray*}
\\
\end{proof}
\begin{thm}
Defining \textup{$z_{s}(t)=k_{s}(t)E^{2-\alpha}$}, The European Call
price under the sub-fractional CEV is given by:

\begin{eqnarray*}
C_{H}\left(S_{0},0\right) & = & \text{e}^{-rT}\int_{-\infty}^{\infty}\max\left\{ S_{T}-E,0\right\} P\left(\left.S_{T},T\right|S_{0},0\right)\text{d}S_{T}\\
 & = & S_{0}Q\left(2z_{s};2+\frac{2}{2-\alpha},2y_{s}\right)-E\text{e}^{-rT}\left[1-Q\left(2y_{s};\frac{2}{2-\alpha},2z_{s}\right)\right]
\end{eqnarray*}
\end{thm}
\begin{proof}
See \cite{araneda2020fractional} and the arguments supplied there.
\end{proof}
Fig. \ref{fig:Price} shows the price of an at-the-money European
Call option under the mixed sub-fractional CEV model with $\gamma=\beta=1$
(blue), in function of the elasticity parameter, for short (3 months,
\ref{fig:=00003D0.25}) and long (2 years, \ref{fig:T=00003D2}) maturities,
with $H=\left\{ 0.5,0.7,0.9\right\} $. Besides, the price under the
mixed fractional CEV (mfCEV) is also drawn (red) by way of comparison.
In all the cases, the sub-fractional pricing is lower than the fractional
one for a fix $H>0.5$. When the Hurst exponent is equal to 0.5, both
models fit with the standard CEV. For $\alpha=2$, the sub-fractional
(fractional) CEV converges to the sub-fractional (fractional) Black-Scholes.

\begin{figure}
\subfloat[$T$=0.25\label{fig:=00003D0.25}]%
{\includegraphics[viewport=30bp 0bp 510bp 420bp,clip,width=0.5\textwidth]{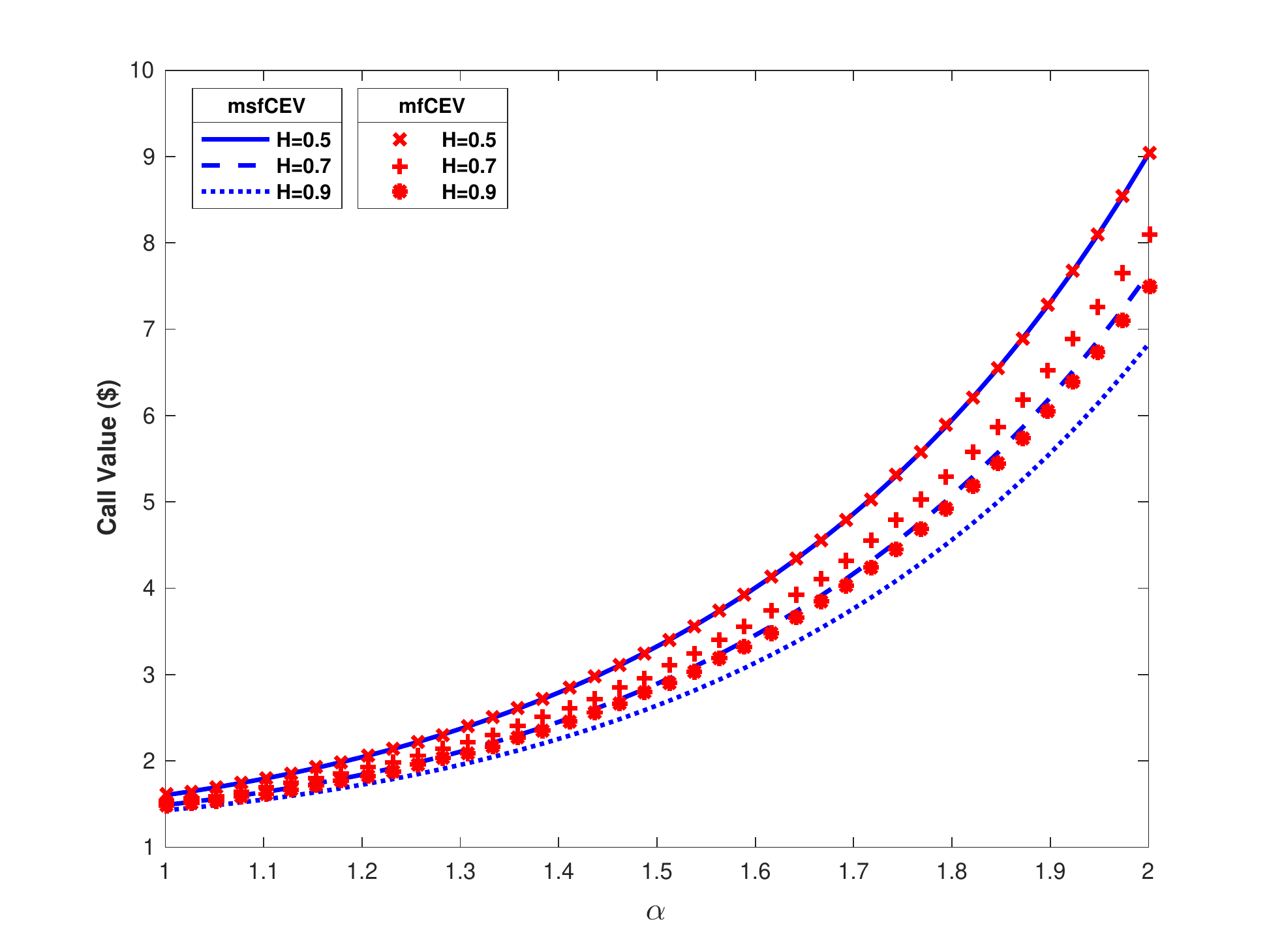}

}%
\subfloat[T=2\label{fig:T=00003D2}]%
{\includegraphics[viewport=30bp 0bp 510bp 420bp,clip,width=0.5\textwidth]{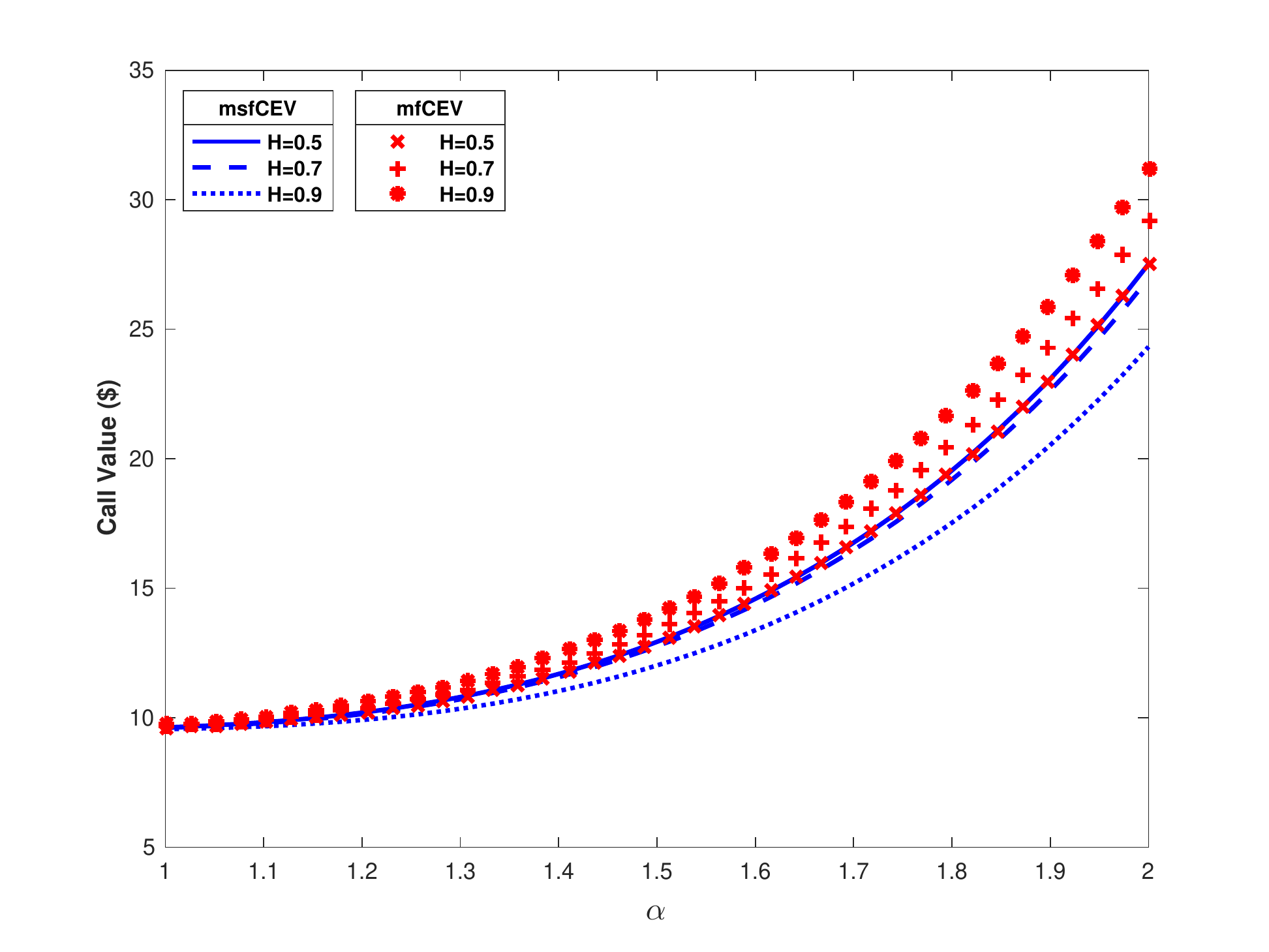}

}

\caption{Price of a European Call option under both fractional and sub-fractional
CEV model as a function of the elasticity ($\alpha)$ and the Hurst
exponent ($H)$. We fix $\sigma=30\%,$ $S_{0}=E=100$ and $r=5$\%.
\label{fig:Price}}
\end{figure}

\section{Numerical Results}

In this section, we test empirically the proposed msfCEV option pricing
model using real option prices data. As an example, use the historical
prices of call options written the S\&P 500. We consider out-of-the-money
and slightly -- at most 5\% -- in-the-money call options with all
available strikes and maturities. Data have been downloaded from OptionMetrics
via Wharton Research Data Services \cite{WRDS}. To put the results
in context, we fit the msfCEV model to option pricing data, but also
we do the same for the standard CEV, mfCEV models, and Black-Scholes
counterparts (Classical, mixed fractional, and mixed sub-fractional);
to compare the performance of each approach. For both mixed fractional
and mixed sub-fractional cases, we have used $\beta=\gamma=1.$

For the empirical analysis we have considered two ways. Fitting jointly,
where the optimal parameters (which minimize the mean squared error
between observed mid-prices and theoretical prices) of the model are
found simultaneously for all the maturities in the data set. On the
other hand, fitting separately, where the optimal parameters are found
for each maturity individually.

Fig. \ref{fig:Performance,-through-the} shows the mean squared errors
(MSE) of the two fitting approaches as a function of the time to maturity
using call prices quoted on January 4, 2010; for both the BS and CEV
model using their classical, mixed fractional, and mixed sub- fractional
versions. First, it\textquoteright s clear that, for the same type
of fit, the CEV approaches outperform the BS counterparts. Second,
when fitting jointly on all maturities, the fractional extensions
outperform the traditional models: the total MSE (summed over all
the maturities) for the CEV is 14.3, while for mfCEV and msfCEV they
are 2.03 and 2.02 respectively. On the BS side, the total MSEs are
18.5 for the classical BS, 11.0 for the mfBS, and 11.1 for the msfBS.
On the other hand, when fitting each maturity separately, we observed
similar performances in the three CEV (total MSE: 1.55 CEV, 1.55 mfCEV,
1.56 msfCEV) and the three BS approaches (albeit at a much higher
total MSE of 10.5). Similar results have also been observed on option
price data quoted on several other days between 2000 and 2019 (not
shown).

In all cases, the mixed fractional and sub-fractional models never
performed worse than the classical ones, often substantially outperforming
them when using the same set of parameters across different strikes
and maturities, i.e. when fitting jointly. Thus, our results confirm
that the CEV model improves the BS model by capturing the volatility
smile across different strikes. Similarly, the proposed msfCEV model
further improves the CEV model by capturing the temporal structure
of option prices across different maturities. Yet, msfCEV and mfCEV
appear very similar in this respect.

\begin{figure}
\includegraphics[width=1\textwidth]{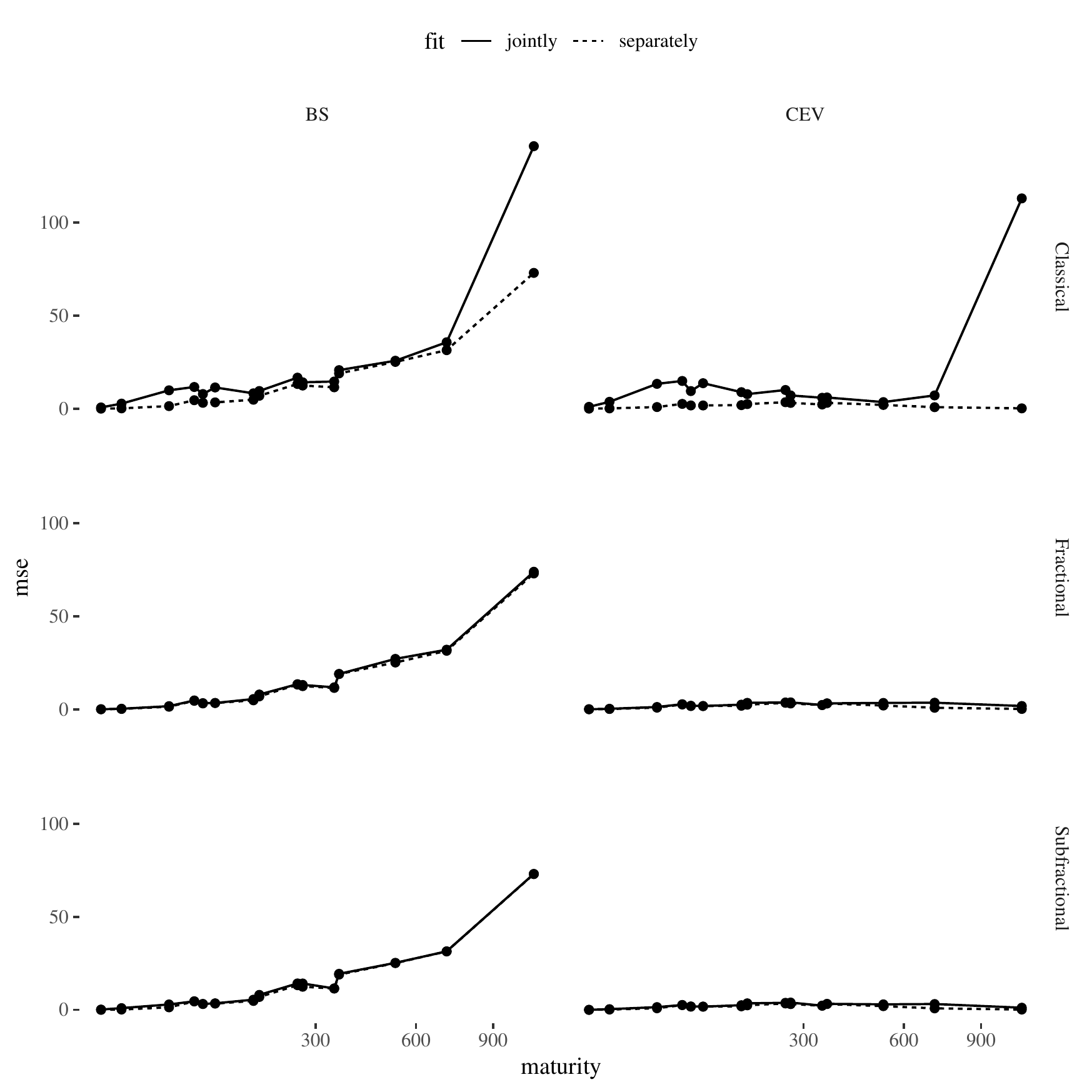}

\caption{Performance, through the mean squared error, under the classical (top),
mixed fractional (middle), and mixed sub-fractional (bottom) approaches
for both Black-Scholes (left) and CEV (right) option pricing models.
\label{fig:Performance,-through-the}}

\end{figure}

\section{Summary}

In this note, we consider the Constant Elasticity of Variance model
driven by a mixed sub-fractional Brownian motion, which allows us
to model financial time-series which present features as a long-range
dependency, non-stationarity increments, leverage effect, and in the
case of the options, the volatility skew pattern. Using sub-fractional
It\^o rules, we drive the Fokker-Planck equation related to the mixed-fractional
CEV process. After getting analytically the transition density function,
the price of a European Call option is obtained in terms of the M-Whittaker
function and the non-central chi-squared complementary function. Empirical
fits on market prices of call options on the S\&P 500 show a clear
improvement over the standard CEV model. Especially the temporal structure
across different maturities is much better captured by the (sub-)fractional
model.

\section{Acknowledgments}

A. Araneda is supported from Operational Programme Research, Development
and Education - Project ``Postdoc2MUNI'' (No. CZ.02.2.69/0.0/0.0/18\_053/0016952).
N. Bertschinger thanks Dr. h. c. Maucher for funding his position. 

\bibliographystyle{unsrt}
\bibliography{15_Users_axelaraneda_Desktop_Research_fBM_Subfractional}

\end{document}